\documentclass[sigconf,screen]{acmart}

\usepackage{balance}
\usepackage{enumitem}
\usepackage{listings}

\AtBeginDocument{%
  }

\acmYear{2025}\copyrightyear{2025}
\setcopyright{rightsretained}
\acmConference[SPAA '25]{37th ACM Symposium on Parallelism in Algorithms and Architectures}{July 28--August 1, 2025}{Portland, OR, USA}
\acmBooktitle{37th ACM Symposium on Parallelism in Algorithms and Architectures (SPAA '25), July 28--August 1, 2025, Portland, OR, USA}
\acmDOI{10.1145/3694906.3743313}
\acmISBN{979-8-4007-1258-6/25/07}

\usepackage{amsmath}
\usepackage{algorithm,algpseudocode}
\usepackage{enumitem}
\usepackage{mathtools}
\usepackage{subcaption}
\usepackage{balance}
\usepackage{listings}

\newcommand{\assign}{\leftarrow}

\newcommand{\E}{\mathbb{E}}

\newcommand{\fsize}{\footnotesize}

\newcommand{\cinput}[1]{--\rm\fsize{\bf Input:} #1}
\newcommand{\coutput}[1]{--\rm\fsize{\bf Output:} #1}
\newcommand{\ceffect}[1]{--\rm\fsize{\bf Effect:} #1}
\newcommand{\cuser}[1]{--\rm\fsize{\bf User Interface:} #1}
\newcommand{\cassume}[1]{--\rm\fsize{\bf Assumes:} #1}
\newcommand{\ccode}[1]{\textsc{#1}}
\newcommand{\depth}{depth}

\mathtoolsset{showonlyrefs=true}

\title{Parallel Batch-Dynamic Maximal Matching\\ with Constant Work per Update}

  \author{Guy E. Blelloch}
  \affiliation{\institution{Carnegie Mellon University}\country{Pittsburgh, USA}}
  \email{guyb@cs.cmu.edu}

  \author{Andrew C. Brady}
  \affiliation{\institution{Carnegie Mellon University}\country{Pittsburgh, USA}}
  \email{acbrady2020@gmail.com}

\newcommand{\myparagraph}[1]{\medskip{\noindent\bfseries\itshape{#1}.~}}



\begin{document}

\begin{abstract}
  We present a work optimal algorithm for parallel fully batch-dynamic
  maximal matching against an oblivious adversary. It processes batches of updates
  (either insertions or deletions of edges) in constant expected
  amortized work per edge update, and in $O(\log^3 m)$ depth per batch
  whp, where $m$ is the maximum number of edges in the graph over
  time.  This greatly improves on the recent result by Ghaffari and
  Trygub (2024) that requires $O(\log^9 m)$ amortized work per
  update and $O(\log^4 m )$ depth per batch, both whp. 

  The algorithm can also be used for parallel batch-dynamic hyperedge maximal matching.  For
  hypergraphs with rank $r$ (maximum cardinality of any edge) the
  algorithm supports batches of updates with $O(r^3)$
  expected amortized work per edge update, and $O(\log^3 m)$ depth per
  batch whp. Ghaffari and Trygub's
  parallel batch-dynamic algorithm on hypergraphs requires $O(r^8 \log^9 m)$ amortized work
  per edge update whp.   We leverage ideas from the
  prior algorithms but introduce substantial new ideas.  Furthermore,
  our algorithm is relatively simple, perhaps even simpler than Assadi and Solomon's (2021)
  sequential dynamic hyperedge algorithm.

  We also present the first work-efficient algorithm for
  parallel static maximal matching on hypergraphs.  For a hypergraph with total
  cardinality $m'$ (i.e., sum over the cardinality of each edge), the
  algorithm runs in $O(m')$ work in expectation and
  $O(\log^2 m)$ depth whp. The algorithm also has some properties
  that allow us to use it as a subroutine in the dynamic algorithm to select
  random edges in the graph to add to the matching.  
  
  With a standard reduction from set cover to hyperedge maximal matching, 
  we give state of the art $r$-approximate static and batch-dynamic parallel set cover algorithms, 
  where $r$ is the maximum frequency of any element, and batch-dynamic updates 
  consist of adding or removing batches of elements.
\end{abstract}

\maketitle

\section{Introduction}

We consider the problem of maintaining the maximal matching of a graph
under batch insertions and deletions of edges (fully dynamic), or
hyperedges in a hypergraph.  There have been several important results
on sequential fully dynamic maximal matching in the past fifteen years.  Onak and
Rubinfeld described an dynamic algorithm with $O(\log^2 n)$ expected
amortized time per update that maintains a constant-approximate maximum
matching~\cite{OR10}.  Baswana, Gupta and Sen (BGS) presented the
first maximal matching algorithm with polylogarithmic time per
update~\cite{BGS11}.  Their algorithm supports single updates
in $O(\log n)$ amortized time with high
probability (whp).\footnote{We say that $f(n)=O(g(n))$ with high probability in $n$ if there exists constants $n_0,c$ such that for all constants $k > 1$, we have for $n \ge n_0$, that $Pr[f(n) \le c k g(n)] \ge 1-n^{-k}$. }  Solomon~\cite{Sol16} then improved the result to
$O(1)$ time per update, again amortized and whp.  These results are
for standard graphs.  Assadi and Solomon (AS) then extended the results to
hypergraphs~\cite{AS21}.  First, they present an algorithm
that, for a hypergraph where each edge has rank at most $r$, supports
insertions and deletions in $O(r^3)$ per update, amortized and whp. 
Then, they improve this algorithm to yield an $O(r^2)$ per update algorithm, also amortized and whp.\footnote{We believe there to be significant gaps in the $O(r^2)$ algorithm, and are unsure of its correctness. We are in contact with the authors.}
All of these results are against an oblivious adversary who knows the 
algorithm but must select the sequence of
updates without knowing any of the random choices by the algorithm.
All of these algorithms maintain a leveling structure that try to keep
vertices or edges on a level appropriate for their degree.

Algorithms that are robust against an adaptive adversary, a stronger adversary who can respond to random choices by the algorithm, are either much less efficient, give a weaker approximation ratio, or only report the approximate size of the matching. 
Liu et al.'s batch-dynamic algorithm has update time $O(\alpha + \log^2 n)$ per element in the batch, where $\alpha$ is the arboricity of the graph ~\cite{Liu2022parallel}. 
Roghani, Saberi, and Wajc (2022) get a better than 2 approximation ratio, but have $O(n^{3/4})$ worst-case update time \cite{RSW22}.\footnote{The approximation ratio is compared to the maximum matching. Maximal matching gives a 2-approximation.}  
Wajc (2020) achieves constant expected time but for $(2+\epsilon)$-approximation \cite{Wajc20}. Bhattacharya et al. (2024) find the size of a matching with approximation ratio better than 2 in polylog update time whp \cite{BKSW24}.
There is no known sequential dynamic maximal matching algorithm that achieves polylog expected updated time against an adaptive adversary. Thus, for the remainder of this work we focus on algorithms that function against an oblivious adversary. 

Ghaffari and Trygub (GT) considered the problem in the parallel batch-dynamic setting~\cite{GT24}.  In this setting, instead of inserting or deleting one edge at a time, the user inserts or deletes, on each step, a batch of edges.  The size of the batch can vary from step to step.  The advantage of considering batches is that an algorithm can then, possibly, apply the updates in parallel, across multiple processors in much better time than applying each one after the other.  Parallel batch dynamic updates have been considered for many graph problems~\cite{Liu2022parallel,acar2019batchconnect,tseng2018batch,acar2020changeprop,dhulipala2021parallel,AndersonBT20,dhulipala2019parallel,AndersonB21}.  In these results, the efficiency of such algorithms is measured in terms of the total work required to process the batch (number of primitive steps), and the \depth{} (also called span or critical path) of dependences in the computation for each batch.  One goal in the design of such parallel batch-dynamic algorithms is to achieve asymptotic work that is as good as the best sequential algorithm, or at least close to as good.  Algorithms that match the sequential work are referred to as \emph{work efficient}.  A second goal is to minimize the \depth{}.  Ideally the \depth{} should be polylogarithmic in the size of the problem.

GT build on the algorithm of BGS to develop a parallel batch-dynamic data structure with polylogarithmic \depth{} per batch update, whp.  However, the work is far from work efficient, requiring $O(r^8 \log^9m)$ work per edge update, whp.  On a regular graph this reduces to $O(\log^9m)$ work per edge update whp.  They ask at the end of their paper, with regards to regular graphs: ``Can we get a parallel dynamic algorithm for maximal matching with $O(1)$ amortized work per update and poly($\log n$) \depth{} for processing any batch of updates?''

In this paper we present a batch-dynamic algorithm that supports updates (insertions or deletions) in $O(1)$ work per edge update for regular graphs and $O(r^3)$ work for hypergraphs, both amortized and in expectation.  The algorithm runs in $O(\log^3 m)$ \depth{} whp.  We therefore answer Ghaffari and Trygub's question in the affirmative (when $r = 2$).
In particular we have:

\begin{theorem}
  There exists a randomized algorithm that solves
  the parallel batch-dynamic maximal matching problem on hypergraphs with $O(r^3)$ amortized
  and expected work per edge update against an oblivious adversary.  Furthermore, processing each batch requires $O(\log^3 m)$
  \depth{} with high probability.  Here $r$ is the rank of the hypergraph and $m$ is the maximum number
  of edges to appear at any point in the update sequence.  The batch size can vary across steps.  \end{theorem}

For $r=2$ (standard graphs), this is constant time per update, which matches the best known sequential algorithm \cite{Sol16}. Thus we can conclude the following corollary. 

\begin{corollary} There exists a parallel randomized batch dynamic graph maximal matching algorithm with constant update time in expectation and polylog \depth{} whp against an oblivious adversary, which is work-optimal. \end{corollary}

Because the batch dynamic and parallel models are more general than the dynamic and sequential models, our algorithm subsumes the previous algorithms in those spaces. Because maximal matching is lower bounded by $\Theta(m)$, our algorithm is optimal for sequential static maximal matching, parallel static maximal matching, and sequential dynamic maximal matching. 

Cardinality matching has various applications \cite{Burq19,Wajc20}. Generally, the vertices of the graph model agents or resources, and edges connect compatible pairs. In a bipartite graph, one focuses on matching resources to consumers, in a general graph, one is pairing agents together, and in a hypergraph, one is selecting compatible groups of agents. Maximal matching accepts a smaller matching size in exchange for time savings. Batch-dynamic updates against an oblivious adversary allow the compatibility preferences of agents to change over time, but only due to outside effects (in ways agnostic to the matching itself). In summary, our setting, parallel batch-dynamic hypergraph maximal matching against an oblivious adversary, addresses the problem of selecting groups of agents quickly, with compatibility preferences that change over time only due to outside effects, where speed is more important than optimizing the number of groups matched. Theoretically, the attention \cite{OR10,BGS11,Sol16,AS21} the sequential dynamic algorithms community has given to maximal matching against an oblivious adversary renders it interesting in the parallel-batch dynamic setting. In turn, matching received attention from the dynamic algorithms community perhaps because of static matching's relevance in TCS as a whole; for example, Edmonds's Blossom algorithm (for finding a maximum matching on a general graph) contributed to the use of polynomial time as a measure of efficiency \cite{Edmonds65,Mehta13}.

Within our algorithm, we use static hypergraph maximal matching as a subprocedure, and our work bound is proportional to the cost of static matching.  However, directly translating the existing static parallel matching algorithms~\cite{Luby86,BirnOSSS13,BFS12} to hypergraphs gives $O(mr^2)$ work, or $O(mr\log m)$ work~\cite{GT24}.  Thus, to improve our work bounds, we develop the first work-efficient, and work optimal, parallel algorithm for maximal hypergraph matching.  This builds on the work of Blelloch, Fineman and Shun (BFS)~\cite{BFS12} and improved bounds by Fischer and Noever (FN) ~\cite{FN20}. These works show that the dependence \depth{} of the simple greedy sequential algorithm for maximal matching is shallow when applied to a random permutation.  This also applies to hypergraph matching.  We extend the approach of BFS to develop a work-efficient algorithm.  The \emph{cardinality} of an edge in a hypergraph is the number of vertices it is incident on, and the \emph{total cardinality} of a hypergraph is the sum of the cardinalities of the edges.

\begin{lemma}
  There exists a parallel algorithm for finding a maximal matching of
  a hypergraph with $m$ edges and total cardinality $m'$ with $O(m')$ expected work and
  $O(\log^2 m)$ \depth{} whp.
\end{lemma}

Interestingly, this static maximal matching algorithm has additional important
properties we make use of.  In particular, because of the random
ordering, each edge is selected to be a match at random among its
remaining incident edges.  For each match, call the remaining
incident edges at the time it is matched, including itself, the sample
space for that match.  We call the size of this sample space the price of the match. 
We show that for any adversarial
(oblivious) sequence of deletions we can pay this price by charging
every edge deleted at or before its incident matched edge at most two units of price in expectation.   The proof is subtle
since the sample spaces themselves depend on the random choices.

We note that $r$-approximate set-cover for instances where each element
belongs to at most $r$ sets reduces to maximal matching in a
hypergraph with rank $r$~\cite{AS21}.  In the reduction, sets
correspond to vertices in the hypergraph and elements to the hyperedges.
Our static and batch dynamic maximal matching algorithms therefore imply the
following corollaries. 

\begin{corollary}
There exists a randomized parallel batch-dynamic
$r$-approximate set cover algorithm that supports batches of
insertions and deletions of elements against an oblivious adversary, each element covered by at most
$r$ sets.  The algorithm does $O(r^3)$ work per element update
(expected and amortized) and has $O(\log^3 m)$ \depth{} per batch whp, where
$m$ is the maximum number of elements to ever appear.  
\end{corollary}

\begin{corollary}
There exists a randomized algorithm for static parallel $r$-approximate set cover, with $O(m')$ expected work and $O(\log^2 m)$ depth whp, where $m'$ is the total cardinality of the $m$ edges in the graph, and where no element appears in more than $r$ sets. 
\end{corollary}

These algorithms are the current parallel static and parallel batch-dynamic state of the art for $r$-approximate set cover. Note that there are two types of set cover approximations, those based on the maximum frequency of any element ($r$), and those based on the maximum size of any set ($\Delta$) \cite{DDLM24}. There are many parallel results with good work bounds, but which sacrifice a $(1+\epsilon)$ factor in the approximation ratio; hypergraph matching is unusual in providing precisely an $r$-approximate cover. Thus, our set cover bounds improve on the bounds given by previous parallel matching algorithms \cite{BFS12, GT24}. For more information on parallel set cover, see Dhulipala et al. \cite{DDLM24}.

For the full version of the paper, see arXiv \cite{BB25a}.  

\subsection{Relationship to Prior Algorithms}
\label{sec:compare}

Our algorithm borrows many ideas from prior level-based dynamic
matching algorithms~\cite{OR10,BGS11,Sol16,AS21,GT24}.  It is most similar to the Assadi and
Solomon~\cite{AS21} hypergraph algorithm.  Like~\cite{AS21}
and~\cite{Sol16}, and unlike the others it is lazy: although a level accurately represents the original sample
size of an edge when settled, the degree of an edge can grow much larger before
it is resettled.   This is what allows it to be work efficient.

All of these prior algorithms have the following pattern. First, observe that inserting edges and deleting unmatched edges is easy; the interesting case is deleting matched edges. By choosing which edges are matched from a large sample space, one can protect them from the oblivious adversary---i.e., one can amortize the expensive cost of deleting the matched edge with the many cheap deletes of unmatched edges in the sample space. To keep track of the sample space associated with a matched vertex, a leveling scheme is kept: higher level means a larger sample space.  Starting with the analysis by BGS ~\cite{BGS11}, all the algorithms distinguish matched edges that are deleted by the user (natural deletions) from those deleted indirectly by the algorithm itself (induced deletions).  The induced deletions need to be charged to the natural ones.

As in previous parallel batch-dynamic algorithms, we use a parallel static maximal matching on insertions to determine which edges to match \cite{Liu2022parallel,GT24}. 
Our algorithm differs from prior algorithms in important ways,
however, and adds several new ideas to achieve parallel work
efficiency.  One of the most important differences is the random
sampling approach.  Instead of selecting a single edge among the
neighbors of a vertex (or in some special cases in~\cite{AS21}, an
edge), it selects the samples using a random greedy maximal matching algorithm across many vertices and
edges.
We show that the randomness in this very
natural algorithm (the trivial greedy maximal matching on a random
ordering of the edges) is sufficient.

An important difference in the analysis of our algorithm is how we
charge induced deletions to natural deletions.  In the prior
algorithms it is done by charging to each newly matched edge some
number of deleted edges on a lower level.  Ultimately these are
charged to a natural deletion.  They therefore carefully have to
ensure that a matched edge only induces deletes on strictly lower
levels.  We instead allow a newly matched edge at a low level to
induce a deleted edge at an arbitrarily higher level. Thus, we only
make an aggregate argument that upper bounds the total size of the
sample space of induced deletions relative to the sample space of
newly created matches.  Ultimately the total sample space of induced
deletions is charged to the sample space of natural deletions.

This leads to another notable change relative to prior hypergraph
algorithms~\cite{AS21,GT24}.  In particular, in their leveling schemes
the levels differ by a factor of $\Theta(r)$.  In our scheme the
levels differ by a factor of $2$.  This is important in our charging
scheme, as discussed in the analysis.  Another difference from the the
lazy schemes~\cite{Sol16,AS21} is that we terminate after at most
$O(\log n)$ ``recursive levels.'' This is important for bounding the
depth.  In Solomon's (2016) algorithm \cite{Sol16}, in certain cases the 
recursion depth is $\Omega(n)$.  It is not clear whether Solomon's 
charging scheme could 
handle terminating after $O(\log n)$ rounds, but it is relatively easy 
to show that our charging scheme does support early termination.

We also differ significantly from Ghaffari and Trygub's parallel 
batch-dynamic algorithm~\cite{GT24}. Their algorithm is based on
BGS instead of Solomon's lazy approach.  It is therefore unlikely that
GT's algorithm can be improved to $O(1)$ work per edge for standard
graphs, even with improvements to all the components.  The GT
algorithm also uses a significantly more complicated method for
randomly selecting edges. In GT, edges are processed by level top-down, 
whereas our maximal matching processes all levels together.
Furthermore, GT has a bespoke method for finding the matching on each
level that involves a triply nested loop.  To ensure that not too many
edges are selected at the same time in parallel, the edges are sampled
at varying density across the iterations of the loops.  The analysis
is quite involved.  Together this leads to the high cost per edge
update.




\section{Preliminaries}

\paragraph{Hypergraph matching}

A \emph{hypergraph} $H = (V,E)$ consists of a set of vertices $V$ and
a set of edges $E \subset 2^V$. Let $n=|V|,m=|E|$ and $m'=\sum_{e \in E} |e|$. The rank of a hypergraph is the maximum
cardinality of any of its edges.  We will be considering hypergraphs of
bounded rank, and will use $r(H)$ or just $r$ to indicate the rank.
We say two edges are \emph{incident} on each other if they share a
vertex, and an edge is incident on a vertex if it includes the vertex
in its set. Two edges are \emph{neighbors} if they are incident on each other. On a hypergraph $H = (V,E)$ a set of edges $M \subset E$
is said to be a \emph{matching} if no two edges in $M$ are incident on
each other.  It is a \emph{maximal matching} if all edges in
$E \setminus M$ are incident on $M$.   We call an edge or vertex 
\emph{free} if it is not incident on any matched edge.

\paragraph{Parallel Model}

We assume the binary forking model with work-\depth{}
analysis~\cite{Blelloch96,blelloch2019optimal}.  In the model, a thread
can fork off any number of child threads which run in parallel, but
asynchronously.  When all finish, control returns to the parent, which
is the join point.  
Such forking can be nested.  Beyond the fork
instruction, each thread executes a standard RAM instruction set.  We
say there is a dependence between two instructions if they are in the
same thread, or one is in an ancestor thread of the other, and
otherwise they are concurrent.  The cost is measured in terms of work
and \depth{}.  The work is the total number of instructions executed by
all threads, and the \depth{}
is the longest chain of dependent instructions.  We allow for
concurrent writes (with constant work per write, and constant \depth{}), but the ordering is non-deterministic.  In
particular, we assume that the values that are read are consistent
with some total order of the instructions that is compatible with the
partial order defined by the dependences.
A PRAM algorithm with $p$ processors and $t$ time can be simulated on
the binary forking model in $O(pt)$ work and $O(t)$ \depth{}.  A binary forking
algorithm with $w$ work and $t$ time can be simulated on the PRAM in
$O(w/p + t \log^* w)$ time randomly w.h.p, or $O(w/p + t \log\log p)$
time deterministically~\cite{blelloch2019optimal}. 
Because mapping from the binary forking model to the PRAM only adds at worst a $O(\log^* w)$ factor to the time in the randomized case, our results are quite robust across the models.

\paragraph{Dynamic model}

We are interested in batch-dynamic maximal matching.  In particular we
support inserting and deleting sets (batches) of edges.  Each batch
can be of arbitrary size.  We also assume we can report for any vertex
its matched edge (or none if the vertex is free) in constant expected time, and that we can return the
set of matched edges in expected work proportional to the size of the
set.  We assume edges have unique identifiers so they can be hashed or
compared for equality in constant time (even though they might have
$r$ endpoints).  This is not an unreasonable assumption since one can
always keep a hash table that hashes an edge to its label
external to our algorithm and then use the labels internally.  As with
much prior work on dynamic maximal
matching~\cite{BGS11,Sol16,AS21,GT24} our algorithms are randomized
and we assume an oblivious adversarial user.  The oblivious adversary does
not know anything about the random choices made by the algorithm.

\paragraph{Standard Algorithms}

We make use of several standard parallel algorithms, which we outline
here.  We make use of prefix sums and filtering, which both require
$O(n)$ work and $O(\log n)$ \depth{}~\cite{blelloch1990pre}.  The filter
operation takes an array with a subset of elements marked and returns
an array with just the marked elements, maintaining order among them.
We will need to generate a random permutation.  For a sequence of
length $n$ this can be done in $O(n)$ expected work and $O(\log n)$
\depth{}~\cite{Gil91a}.

We will make significant use of semisorting~\cite{Valiant91}.  The
semisorting problem is given a set of keys, to organize them so equal
keys are adjacent in the output, but not necessarily sorted.  For
semisorting we assume we have a hash function on the keys.
Semisorting a collection of $n$ keys requires $O(n)$ expected work and
$O(\log n)$ \depth{} whp ~\cite{Valiant91,GSSB15}.  Semisort can, in turn, be
used for several other operations.  In particular we will need a
\emph{groupBy} function that takes a set of key-value pairs and groups
all values with the same key into its own subset.  It returns a set of
pairs, each consisting of a unique key and all the values with that
key.  This is easily implemented by a semisort followed by a prefix
sum to partition into subsets.  We will also need a \emph{sumBy}
operation that takes a set of key-value pairs where the values are
integers, and for each unique key sums up the values with that key.
It returns a set of key-value pairs where the keys are unique and the
values are the corresponding sums.  This again can be implemented with
a semisort. We also need a \emph{removeDuplicates} function, that given a 
list with repeated elements, returns the unique list. This can be implemented with 
a groupBy. In our pseudocode, when we use set notation $\{\}$, we are 
implicitly calling removeDuplicates as needed.

We also make use of dictionary data structures based on hashing.  A
parallel dictionary structure supports batches of insertions,
deletions, and membership queries.  A batch of any of these operations
$k$ can be implemented in $O(k)$ expected work and $O(\log^* k)$ \depth{} whp in $k$ (which is $O(\log n)$ \depth{} whp in $n$) ~\cite{Gil91a}.  This assumes the dictionary is large enough to
hold the elements (i.e., the load factor is constant).  We will need
growable/shrinkable dictionaries, but given that we are happy with
amortized bounds this can be implemented using the standard
doubling/halving tricks (i.e., when too full, double the size and copy
the old contents into the new, and when too empty, half the size and
copy).  In our analysis we will assume that batch of $k$ operations
can be applied to a dictionary of size $n$ in $O(k)$ expected amortized
work and $O(\log (n+k))$ \depth{} whp.  Also beyond membership queries we
will need to extract a sequence of all the current elements.  This can
be performed with $O(n)$ expected work and $O(\log (n+k))$ \depth{} whp.  The size of each
dictionary is $O(n)$. 

Finally we will need a \ccode{findNext} function that, given an index
$i$ in an array, finds the next index $j$ that satisfies some
predicate.  It needs to run in $O(j -i)$ work and $O(\log (j-i))$
\depth{}.  This can be implemented using doubling then binary search.  On each round $k$,
search the next $2^k$ elements.  In a round, in parallel each element in 
the range checks the predicate and writes to a flag if
satisfied.\footnote{We can do this because we support concurrent writes with a non-deterministic ordering.} If none is found in a round, then go to the next round.
Each round takes $O(2^k)$ work and $O(1)$ \depth{}. 
When a round succeeds, we then apply binary search over that range.
In each step of binary search, we check if the left half of the list contains an
 alive element by concurrently writing to a flag. If the left half contains
  an alive element, we recurse on the left half; otherwise, we recurse on 
  the right half. Thus the index of the next element that satisfies the predicate 
  can be identified in $O(2^k)$ work and $O(k)$ \depth{}.  The total work 
  and \depth{} is therefore within the required bounds.
  
\paragraph{Parallel insertions, deletions and increments}
In our pseudocode, for simplicity, we often insert and delete elements
from various sets and maps in parallel loops where the parallel iterations
might be on the same key.  In our parallel model, this would require
first applying a groupBy to gather elements based on the set they are
updating, and then applying a batch update to each such set.
Similarly, in the \ccode{updateTop} code in Section
\ref{Section:RGMM}, we increment various counters in parallel.  This
can be implemented with a sumBy.  Both of these require at most
expected amortized work proportional to the number of updates, and
logarithmic \depth{} whp.  The amortization comes from possibly
needing to grow or shrink some of the hash tables for the sets.

\section{Random Maximal Hyperedge Matching \label{Section:RGMM}}

In our dynamic algorithm, we must find a maximal matching on a
subgraph in two situations: when inserting edges and during random
settles.  When inserting edges, any fast maximal matching would
suffice.  However, for the random settles, we need, among other
properties, for the matched edges to be picked with sufficient
randomness to protect them from the oblivious adversary.  For brevity
we use the same algorithm for both cases, and in particular we develop
an efficient parallel implementation of the static sequential random
greedy algorithm for this purpose.

The sequential random greedy maximal matching algorithm first randomly
permutes the edges of a graph, and then passes through them once in
the random order, selecting an edge if it has not been deleted, and if
it is selected deletes all incident edges.  We refer to the order in
the permutation as the \emph{priority} (highest first).  The output of this algorithm is known as the the lexiographically first maximal matching \cite{Cook85}. For our
purposes, it is important to keep track of the edges that are deleted
when a matched edge is selected.  For a matched edge $e$ we refer to
these edges as the \emph{sample space} $S_e$ of $e$. The sample space
includes $e$ itself.  Note that the sample spaces of the matched edges
form a partitioning of the original edges since every edge is deleted
exactly once.  Figure~\ref{fig:seqmm} gives pseudocode for
the algorithm.

\begin{figure}
\begin{lstlisting}
@\cinput{A Graph $G=(V,E)$}@
@\coutput{A matching augmented with sample spaces}@
sequentialGreedyMatch$(G=(V,E))$:
  $\pi \assign$ getRandomPermutation()
  $E' \assign  E \mbox{ sorted by } \pi$
  $X \assign \emptyset$
  for $e \in E'$:
    if $\mbox{free}(e)$:
      $\mbox{free}(e) \assign false$, $S_e \assign  \{e\}$
      for $\mbox{edges } e' \mbox{ incident to } e$ :
        if $\mbox{free}(e')$ :
          $\mbox{free}(e') \assign false$, $S_e \assign  S_e \cup \{e'\}$
      $X \assign  X \cup \{(e,S_e)\}$
  return $X$
\end{lstlisting}
\caption{Sequential random greedy maximal matching.}
\label{fig:seqmm}
\end{figure}

Greedy maximal matching can be implemented to run in parallel,
returning the identical output.  We call an edge $e$ a \emph{root} of
a graph if all remaining edges incident to $e$ have later priority
than $e$.  In each round of the parallel algorithm, we add all roots
($W$) to the matching, remove the matched edges along with their
incident edges, and update the root set.  Since edges in the root set
can have mutually neighboring edges, to mimic the sample spaces of the
sequential algorithm, we assign such a contested edge to the sample
space of the highest priority incident root.  Blelloch, Fineman and
Shun (BFS) show that this procedure takes $O(\log^2 n)$ rounds whp
\cite{BFS12}; and the analysis was improved to $O(\log n)$ rounds whp
by Fischer and Noever (FN) \cite{FN20}. Both of these results reduce matching
to finding a maximally independent set of edges, where independence is with respect to edge incidence.  For
$r = 2$, BFS also provide a work efficient (linear work) parallel
algorithm for this approach that runs in polylogarithmic depth.  A
direct translation of this algorithm to hypergraphs, however, is not
work efficient.

Here we develop a work-efficient maximal matching algorithm on
hypergraphs by modifying the method by which the root set is updated
to yield $O(m')$ work, where $m'$ is the total cardinality of the
hypergraph.  Towards this end we keep track of the following data
structures.  
For each vertex $v$, we store the edges incident on $v$ in two data structures. 
We maintain the remaining edges adjacent to $v$,
$N(v)$, as an unordered set; when an edge incident to $v$ is removed from the graph, we remove this edge from $N(v)$.  
We also maintain \ccode{edges}$(v)$, an
array of edges sorted by priority, and \ccode{top}$(v)$, the index of the
highest priority remaining edge in \ccode{edges}$(v)$.
We do not remove edges from \ccode{edges}$(v)$;
instead, we increase \ccode{top}$(v)$ to indicate that prior edges are
deleted. Each edge $e$ has an atomic counter
(\ccode{counter}), initially set to zero, indicating on how many
vertices it is the highest priority remaining edge.  We set
\ccode{done}$(e)$ to be true if $e$ is being deleted this round, and
false otherwise.

The algorithm is described in Figure~\ref{fig:parallelmm}.  The root
set is maintained in $W$ and each round is given by an iteration of
the while loop.  The first \ccode{groupBy} in the while loop gathers, for
each edge adjacent to the root set, the neighboring roots.  The
second \ccode{groupBy} then gathers for each root $w$ its sample space
(i.e., the neighbors $e$ of $w$ for which $w$ has the highest priority
among $e$'s neighboring roots).  We now have to delete the edges in
$W$ and its neighbors (together \ccode{finished}), update the
appropriate data structures, and generate the new root set.  Note that
an edge is a root iff it is the highest priority remaining edge with
respect to all of its vertices. When
\ccode{edges}$(v)$[\ccode{top}$(v)$] is set to be deleted, we call
\ccode{updateTop}$(v)$, which finds $v$'s new highest remaining edge
$e'$ and increments \ccode{counter}$(e')$. If \ccode{counter}$(e')$
equals the rank of $e'$, we add it to the new root set. We repeat on
the new root set and remaining graph.

\begin{figure}
  \begin{lstlisting}
@\cinput{A vertex $v$}@
@\ceffect{Corrects top$(v)$ if needed}@
@\coutput{Returns edge to add to frontier if needed}@
updateTop(v) :
  if $\mbox{done}(\mbox{edges}(v)[\mbox{top}(v)])$ : 
    $\mbox{top}(v)\assign \mbox{findNext}(\mbox{top}(v), \mbox{edges}(v))$
    if $\mbox{top}(v) \neq |\mbox{edges}(v)|$ :
      $e_t = \mbox{edges}(v)[\mbox{top}(v)]$
      $\mbox{increment}(\mbox{counter}(e_t))$
      if $\mbox{counter}(e_t) = |V(e_t)|$ : return $e_t$
  return $\bot$

@\cinput{A Graph $G=(V,E)$}@
@\ceffect{Clears $G$}@
@\coutput{A matching augmented with sample spaces}@
parallelGreedyMatch($G=(V,E)$):
  $\pi \assign$ getRandomPermutation()
  parfor $v \in V$ :
    $\mbox{edges}(v) \assign$ sort $\{e\; |\; v \in e\} \mbox{ by } \pi(e)$
    $\mbox{top}(v) \assign 0$
    $\mbox{increment}(\mbox{counter}(\mbox{edges}(v)[0]))$
  $W \assign \{e \in E \;|\; \mbox{counter}(e)=|V(e)|\}$
  $X \assign \emptyset$
  while (|W| > 0) :
    $D \assign \mbox{groupBy}(\{ (e, w) : w \in W, e \in N(V(w)) \})$
    $X' \assign \mbox{groupBy}(\{(\arg \min_{e \in F} \pi(e)), e) : (e, F) \in D\})$  
    $X \assign X \cup X'$
    $\mbox{finished} \assign W \cup N(V(W))$
    parfor $e \in \mbox{finished}$ : $\mbox{done}(e) \assign true$
    $V_f \assign \bigcup_{e \in \mbox{finished}} V(e)$
    $W' \assign  \{\mbox{updateTop}(v): v \in V_f\}$
    $W \assign  \{e \in W' \;|\; e \neq \bot\}$
    parfor $e \in \mbox{finished}, v \in V(e)$ : $\mbox{delete}(N(v), e)$
  return $X$
\end{lstlisting}
  \caption{Parallel random greedy maximal matching.  In the code, $V(e)$ indicates the vertices
  incident on $e$, and for a set of vertices $V$, $N(V)$ indicates the union of all edges incident on at least one vertex in $V$.}
  \label{fig:parallelmm}
\end{figure}

We now analyze the cost of \ccode{parallelGreedyMatch}.

\begin{lemma} All of the calls to \ccode{updateTop} cost, in aggregate, $O(m')$ work, where $m'$ is the total cardinality of the edges and $m$ is the total number of edges. Each call individually is $O(\log m)$ \depth{}. \end{lemma}

\begin{proof}
Each call to \ccode{updateTop}$(v)$ slides $v$'s pointer along its edge list to the right, to the new highest priority remaining edge. Using the findNext operation, a call that slides a pointer $d$ distance costs $O(d)$. Because there are $m'$ total edges across the lists, the pointers slide a total of $O(m')$ distance, for $O(m')$ work. Since the edge lists have at most $m$ members, each findNext call takes $O(\log m)$ \depth{}, so $\ccode{updateTop}$ has $O(\log m)$ \depth{}. \end{proof}

\begin{theorem} \ccode{parallelGreedyMatch} takes $O(m')$ expected work and $O(\log^2 m)$ \depth{} whp.
\label{lemma:greedymatchcost}\end{theorem}
\begin{proof}  Generating the permutation for the edges takes linear work. Because the edges are given random priorities, bucket sorting takes $O(m)$ work in expectation and $O(\log m)$ \depth{} whp \cite{BFS12}\cite[p. 200-204]{CLRS}. Each edge $e$ is deleted once and costs $|e|$ work outside of \ccode{updateTop}, for $O(m')$ work. In aggregate, internal calls to \ccode{groupBy} cost no more than $O(m)$ in expectation. Since \ccode{updateTop} in aggregate costs $O(m')$, we have $O(m')$ work total. Since \ccode{updateTop} and filter have $O(\log m)$ \depth{}, each round of GreedyMatch is $O(\log m)$ \depth{}. By FN there are $O(\log m)$ rounds whp \cite{FN20}, for $O(\log^2 m)$ \depth{} total. \end{proof}


\subsection{Random Selection in Greedy Maximal Matching}
\label{Section:3.1}


In Solomon's algorithm \cite{Sol16}, about half of the sampled edges get deleted before the matched edge, because the matched edge was chosen uniformly at random from the sample set. This property is crucial for the analysis, because it shows that an expensive matched edge deletion can be amortized against many cheap updates (the unmatched deletions before it). In Solomon's algorithm, the sample space is chosen deterministically (based on the present graph state), then a random mate is chosen. 

This is not true of \ccode{parallelGreedyMatch}, which chooses a matched edge and its sample space together, in the same random process. Note that we do not tell the adversary the sample space partitioning of a graph. But if we did tell the adversary the sample space of a matched edge, because the adversary knows the structure of the graph, the adversary would learn something about which edge was likely matched. For example, consider the path on 4 vertices $(1,2), (2,3), (3,4)$. Suppose that \ccode{parallelGreedyMatch} is run and the sample space of one of the matched edges is $\{(1,2),(2,3),(3,4)\}$. Since a matched edge is incident on all its sample space it follows that $(2,3)$ was matched. Thus, even knowing the size of the sample space reveals information about which edge was likely matched. We still can demonstrate that \ccode{parallelGreedyMatch} has the necessary randomness to protect its matched edges from the adversary, but we must approach the analysis differently. Note that the analysis presented in this section is independent of the dynamic matching that follows. In particular, in this setting, induced and natural epochs do not exist. We use the results from this section as a black box in Section \ref{sec:overall-cost} to show bounds on the dynamic algorithm.

We start by presenting an alternative simple proof for the case where
the sample space is chosen deterministically, and then extend to the
general case.\footnote{This proof seems significantly simpler than the
proofs in~\cite{Sol16,AS21}.}  To each matched edge $e \in E$, we
assign price $|S_e|$, which has to be paid for. Unmatched edges
are given price 0.  In the dynamic matching setting, after an element
is chosen to be matched from the sample space, the user deletes one
edge $d_t$ per time step $t$ until all of the edges of the sample
space are deleted.\footnote{The amortized analysis assumes the graph
ends up empty.}  The user does these deletes in an arbitrary and
adversarial order, but since the user is oblivious, she does not know
which edge in the sample space was matched.  Let $p(e)$ indicate the
matched edge for which $e$ is in its sample space ($S_{p(e)}$), allowing us later
to consider multiple matched edges and corresponding sample spaces.

We say that when the user deletes an edge $d_t$ after $p(d_t)$ was
deleted, that the deletion of $d_t$ is \emph{late}, and otherwise
it is \emph{early}.  Note that if $d_t = p(d_t)$, then it is early.
We need to pay the price for a matched edge with the user
deletes on its sample space.  We cannot, however, use late deletes to
pay the price since after the match is deleted a new incident edge is
matched and the future delete of an edge might have to pay for that
one.  We therefore charge all the price of a match to early deletes on
edges in its sample space, and will argue that with an expected payment
of at most 2 units per delete we can cover the full price.

We use the following charging scheme.  When the user deletes an early
unmatched edge $e$, it pays 1 unit to subtract 1 unit from
$p(e)$'s price.  When the user deletes a matched edge $e$, which must
be early, it pays the remaining price of $e$---i.e., the original
price minus the number of edges in its sample space that have been
visited.  If $e$ is late it pays nothing.  Note that in this scheme
once all edges have been deleted the payment of early deletions fully
covers the price.  Let $\Phi(e)$ be the random variable for the payment
made when deleting the edge, and $U(e)$ be the set of edges in the sample set
of $e$ just before it is deleted.  We are interested in $\E[\Phi(e)]$
for early deletions.

\begin{proof} \emph{(For deterministically chosen sample spaces, if $d_t$ is an early delete, then $\E[\Phi(d_t)] < 2$.)}
  Given the two cases for costs for early deletes, we have
that $$E[\Phi(e)] = p_{\mbox{\small unmatched}} \times 1 +
p_{\mbox{\small matched}} \times |U(e)|\; .$$ However $p_{\mbox{\small
    matched}}$ is exactly $1/|U(e)|$, since each remaining edge is
equally likely to have been picked from the sample space.  Hence we
have $E[\Phi(e)] < 1 \times 1 + \frac{1}{|U(e)|} \times |U(e)| = 2$.
\end{proof}

Note that to extend this argument to our more complicated situation
where the sample space is itself randomized, the first term remains
the same since the probability is always upper bounded by 1.  The
second term, however, is more complicated since ``each remaining edge is equally
likely to have been picked from the sample space'' is not well defined
since the sample space depends on the random permutation in our
matching.  Intuitively, it still works because the possible
permutations can be partitioned into equivalence classes each
consisting of a fixed sample space, and within each equivalence class
we can argue that every element in the sample space is equally likely.
We formalize this idea in the following theorem and proof.


\begin{lemma} \label{Lemma:3.4}
For sample spaces created by \ccode{sequentialGreedyMatch},
if $d_t$ is an early delete, then $\E[\Phi(d_t)] \le 2$.
\end{lemma}

\begin{proof}
Let $\Pi$ be the set of permutations (that we would input into
\ccode{sequentialGreedyMatch}) such that $d_t$ is an early
delete. Then $\Pi$ is the set of permutations such that $p(d_t) \in
U(d_t)$ (henceforth we will use $U$ for $U(d_t)$).  Let $r$ be the time
step in \ccode{sequentialGreedyMatch} at which $d_t$ is marked to be
not free (the time step where $d_t$ is effectively removed from the
graph).\footnote{Note that we have two distinct uses of time: the time
in the user delete sequence, and the time in the run of
\ccode{sequentialGreedyMatch}. Note that the
\ccode{sequentialGreedyMatch} time happens completely before the user
delete sequence time. Note that these times are separate from each
other: if $t$ is large, it is possible that $r$ is large or small.}
Let $P$ be the set of edges incident on $d_t$ at step $r$ (the
remaining edges after permutation deletes). Group $\Pi$ into
equivalence classes based on, for $\pi \in \Pi$, the value of $[r,
  (\pi_i)_{i=1}^{r-1}]$. Note that $P$ is determined by $[r,
  (\pi_i)_{i=1}^{r-1}]$. Let $\pi \in \Pi$, and let $B$ be the
equivalence class containing $\pi$. 

We consider the possible values of $\pi^{-1}(r)$. Since $d_t$ must
delete in step $r$, $\pi^{-1}(r)$ must be in $P$. Note that
$\pi^{-1}(r) =p(d_t) \in U$. Since there are no other restrictions, we
have that $\pi^{-1}(r) \in P \cap U$.  Thus $|B|=|P \cap U| (m-r)!$.
Next, we consider the value of $\Phi(d_t)$ if $d_t$ is matched. At
user time 0, $d_t$ had price $|P|$. However, each unmatched edge in
$d_t$'s sample space that the user visits reduces $d_t$'s price by
1. Thus, the price of $d_t$ is $|P|-|P \cap \overline{U}|=|P \cap
U|$. There are $(m-r)!$ permutations in $B$ where $d_t$ matches at
time step $r$.

Drawing a permutation uniformly at random from $B$, we
get \begin{multline*} \E[\Phi(d_t)]=(1-\frac{1}{|P \cap U|})1 + \frac{1}{|P \cap U|} |P \cap U|
   \le 2. \end{multline*} Since $B$ was an
arbitrary equivalence class we can conclude that $\E[\phi(d_t)] \le
\max_{B} \E[\phi(d_t) | \pi \in B] \le 2$.
\end{proof}

Note that Lemma \ref{Lemma:3.4} is true for all times $t$; this means that, regardless of the deletes the user chooses, the next delete has constant expected payment. Even though the payment given by later deletes depends on the payment given by earlier deletes, the expectation remains no more than 2. We also note that, when a graph is fully deleted, the sum of the payments of early deletes is the number of edges in the graph. 

Let $\Phi'(d_t)=\begin{cases} \Phi(d_t) & d_t \textrm{ is early} \\ 0 & \textrm{otherwise} \end{cases}$. 
\begin{lemma}\label{Lemma:3.5} The early deletes on the sample space of deleted matched edge $e$ contribute $|S_e|$ total payment. Furthermore, $\sum_{i=1}^m \Phi'(d_t)=m$. \end{lemma}

\begin{proof} 
Observe that if a matched edge $e$ has price $|S_e|-j$ on deletion, then $j$ unmatched edges in $S_e$ were deleted, and the updates deleting these $j$ unmatched edges paid $j$. Thus the early deletes on the sample space of $e$ contribute $|S_e|$ total payment. Thus, all of these early deletes pay $m$ in total. \end{proof}

\newcommand{\gap}{2}

\begin{figure*}
  \small
    \begin{minipage}[t]{.47\textwidth}
    \begin{lstlisting}
isHeavy$(e)$ : return $|C(e)| \geq 4r^2 \gap^{l(e)}$
      
@\cinput{an edge $m$ along with a set of incident edges $S_e$}@
@\ceffect{$m$ is added as a match with sample edges $S_e$}@
addMatch$(m, S_e)$ :
  $\mbox{insert}(M,m)$
  parfor $e \in S_e$ :
    type$(e) \assign$sampled
    $p(e) \assign m$
  type$(m) \assign$matched
  parfor $v \in V(m)$ : $p(v) \assign m$
  $S(m) \assign S_e$
  $C(m) \assign \emptyset$

@\cinput{a matched edge $m$}@
@\cassume{sample edges are converted to cross edges}@
@\ceffect{removes $m$'s owned edges and frees vertices}@
@\coutput{the cross edges owned by $m$}@
removeMatch$(m)$ :
  $\mbox{delete}(M,m)$
  $E \assign C(m)$
  parfor $v \in V(m)$ : $p(v) \assign \bot$
  parfor $e \in E$ : removeCrossEdge$(e)$
  return $E$

@\cinput{an edge $e$ that is not in the structure}@
@\ceffect{add $e$ to the structure as a cross edge}@
addCrossEdge$(e)$ :
  type$(e) \assign$cross  
  $v_{max} \assign argmax_{v \in V(e)} l(p(v))$
  $e' \assign p(v_{max})$
  insert$(C(e'), e)$
  parfor $v \in V(e)$ : insert$(P(v, l(e')), e)$ 

@\cinput{a cross edge $e$}@
@\ceffect{removes $e$ from the structure}@  
removeCrossEdge$(e)$ :
  $\mbox{delete}(C(p(e)),e)$
  parfor $v \in V(e)$ : delete$(P(v, l(p(e))), e)$
  type$(e) \assign$unsettled

@\cinput{a set $E$ of newly matched edges}@
@\ceffect{adjusts cross edge incident on $E$ to their new level}@
adjustCrossEdges$(E)$ :
  $V \assign \bigcup_{e \in E} V(e)$
  $C \assign \bigcup_{v \in V} \bigcup_{i \in [0,l(p(v)))} P(v,i)$
  parfor $e \in C$ : removeCrossEdge$(e)$
  parfor $e \in C$ : addCrossEdge$(e)$
  
     \end{lstlisting}
   \end{minipage}  \hspace{.2in}
   \begin{minipage}[t]{.47\textwidth}
     \begin{lstlisting}
@\cinput{the cross edges owned by removed heavy matches}@
@\ceffect{remove and rematch possibly breaking matches}@
@\coutput{the cross edges owned by new removed matches}@
randomSettle$(E)$ :
  $G' \assign (V(E),E)$
  $X \assign$parallelGreedyMatch$(G')$
  stolen $\assign \{p(v) : (m,\_) \in X\; |\; \exists_{v \in V(m)} (p(v) \neq \bot)\}$
  parfor $(e, S_e) \in X$ :  addMatch$(e, S_e)$
  adjustCrossEdges$(\{m : (m,\_) \in X\})$
  bloated $\assign \{m : (m,\_) \in X\; | \; \mbox{isHeavy}(m)\}$
  return deleteMatchedEdges$(\mbox{bloated} \cup \mbox{stolen})$

@\cinput{a set of matched edges}@
@\ceffect{remove matches, rematch the light matched edges}@
@\coutput{the cross edges owned by the heavy matches}@
deleteMatchedEdges$(E)$ :
  parfor $e \in \bigcup_{m \in E} S(m)$ : addCrossEdge$(e)$
  heavy $\assign\{m \in E\; |\; \mbox{isHeavy}(m)\}$
  light $\assign E \setminus \mbox{heavy}$
  insertEdges$(\bigcup_{m \in \mbox{light}} \mbox{removeMatch}(m))$
  return $\bigcup_{m \in \mbox{heavy}} \mbox{removeMatch}(m)$

@\cuser{insert a batch of edges}@
insertEdges$(E)$ :
  free $\assign \{e \in E\; |\; \forall_{v \in e} (p(v) = \bot)\}$
  $G' \assign (V(free),free)$
  $Y \assign \{e : (e, \_) \in \mbox{parallelGreedyMatch(G')}\}$
  parfor $e \in Y$ : addMatch$(e, \{e\})$
  parfor $e \in E \setminus Y$ : addCrossEdge$(e)$
 
@\cuser{delete a batch of edges}@    
deleteEdges$(E)$ :
  matched $\assign\{e \in E \; | \; \mbox{type}(e) = \mbox{matched}\}$
  parfor $e \in E \setminus \mbox{matched}$ :
    if $\mbox{type}(e) = \mbox{cross}$ : removeCrossEdge$(e)$
    else : $\mbox{delete}(S(p(e)), e)$
  parfor $m \in \mbox{matched}$ : delete$(S(m), m)$
  $E' \assign$deleteMatchedEdges(matched)
  sampledEdges $\assign 0$
  while $(2 |E'| > \mbox{sampledEdges})$ :
    sampledEdges $\assign \mbox{sampledEdges} + |E'|$
    $E' \assign$randomSettle$(E')$
  insertEdges$(E')$
    \end{lstlisting}
   \end{minipage}
   \caption{Parallel batch-dynamic algorithm for maximal matching with $O(r^3)$ expected amortized cost
     per edge update.}
   \label{fig:parallelhyperedge}
\end{figure*}

\begin{table}
  \begin{tabular}{|c|l|}
    \hline
    \ccode{type}$(e)$ & one of \it $\{$matched, sample, cross, unsettled$\}$\\
    $p(e)$ & owner of edge $e$\\
    $V(e)$ & vertices of edge $e$\\
    $S(m)$ & sample edges owned by matched edge $m$\\
    $C(m)$ & cross edges owned by matched edge $m$\\
    $l(m) = \lfloor \lg |S(m)| \rfloor$ & level of matched edge $m$\\
    $p(v)$ & matched edge incident on $v$, or $\bot$ if none\\
    $P(v,l)$ & $\{e \in \mbox{cross edges}\; |\; p(e) = l \wedge v \in e\}$\\
     \hline
  \end{tabular}
  \caption{The notation used in the algorithm.}
  \label{table:symbols}
\end{table}

\begin{figure}
\includegraphics[width=\linewidth]{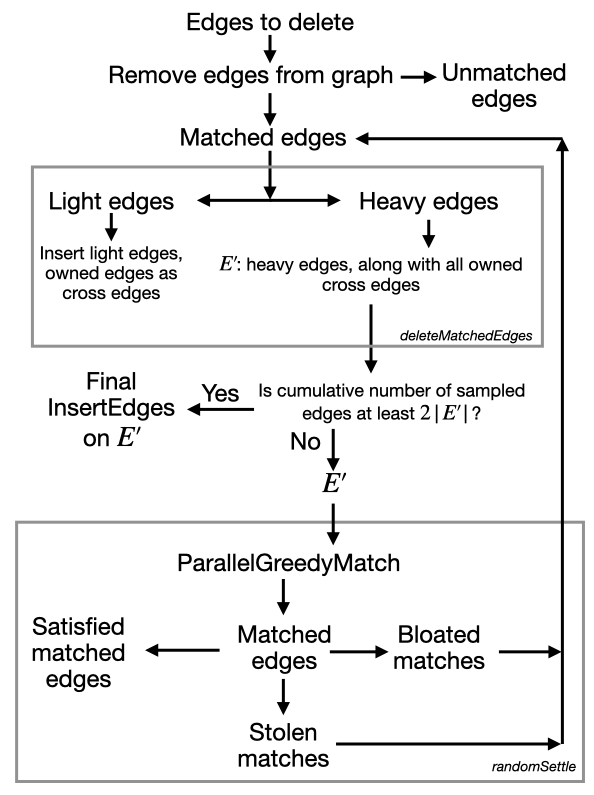}
\caption{Flow chart representation of \ccode{deleteEdges}. After edges are removed from the graph, they are separated into matched and unmatched. The matched edges are passed into \ccode{deleteMatchedEdges}, where they are separated into light and heavy. The light matches (and their cross edges) are reinserted into the graph, while the heavy matches and their cross edges ($E'$) are sent to \ccode{randomSettle}. A static maximal matching is run on $E'$, and the matches are separated into satisfied (no further action needed), stolen, and bloated. We repeat the process by sending the stolen and bloated edges to $\ccode{deleteMatchedEdges}$. We repeat until we have sampled sufficiently many edges.}
\end{figure}

\section{Algorithm}

We now describe our parallel batch-dynamic algorithm.
Roughly, our algorithm
follows the structure of prior level-based dynamic maximal matching
algorithms~\cite{BGS11,Sol16,AS21,GT24}.  In these algorithms, and
ours, inserting edges and deleting unmatched edges is reasonably
straightforward, and the difficulty comes in deleting matched edges.
When deleting a matched edge we need to find if any of the incident
edges can be matched.  If the degree of the edge is high this can be
costly.  The idea of the prior algorithms is to randomly select a
possible edge to match across a sufficiently large sample (roughly its
degree).  The oblivious adversary therefore does not know which of the
sampled edges is selected and in expectation will delete half the
edges before reaching the one that was selected.  When a matched edge
is encountered, the algorithms, roughly speaking, need to
explore neighboring edges, but this cost is amortized against the
prior cheap deletes of unmatched edges.

The level-based algorithms all use some leveling scheme that consists
of a logarithmic number of levels and places every matched edge at a
level that is appropriate for its sample size (and degree).  The algorithms
then ensured that the cost of finding matches in incident edges is
either proportional (within factors of $r$) to the sample size for
that level, or it is possible to re-sample among a larger set and
charge the cost to the new sample. Finding an edge to add to a 
matching is referred to as \emph{settling}.

We now describe how our leveling scheme works and outline the algorithm
and analysis.  We will describe the analysis in detail in Section~\ref{sec:analysis}.
Our leveling structure is defined as follows.

\begin{definition}
  A \emph{leveled matching structure} 
for an undirected hypergraph $G = (V,E)$ and matching $M \subseteq E$
maintains the following invariants:
\label{ref:leveled}
\begin{enumerate}[leftmargin=*]
\item All edges of $E$ are either \emph{cross edges} or \emph{sampled edges}. $M$ is a subset of the sampled edges.
  \label{invariant1}
\item Every edge is owned by an incident matched edge (a matched edge
  owns itself).
  \label{invariant2}
\item When an edge $e$ becomes matched, let $s$ be the number of sample edges $e$ owns. We assign $e$ the level $l = \lfloor \lg s \rfloor$.
  \label{invariant3}
\item The owner $m$ of any cross edge $e$ must be on the maximum level of any matched
  edges incident on $e$.
  \label{invariant4}
\end{enumerate}
\end{definition}
The second invariant implies the matching is maximal.  Throughout we use $\lg n$ to mean $\log_2 n$.

The invariants above are maintained between every batch operation.
During the operation itself, some can be violated.  In particular,
updates will involve ``unsettling'' edges to be resettled.  

To efficiently access the leveled matching structure we will use the
following data structures.
\begin{itemize}[leftmargin=*]
\item
  We maintain the set of matched edges $M$.
\item
  For every edge $e$ we maintain its vertices ($V(e)$), the type of the
  edge (\ccode{type}$(e) \in \{$\textit{matched, sample, cross,
    unsettled}$\}$), and its owner ($p(e)$).
\item
  For every matched edge $m$ we further maintain a set of its sampled
  edges ($S(m)$) and a set of its owned cross edges ($C(m)$).

\item
  For every vertex $v$ we maintain $p(v)$ that specifies the matched
  edge that covers it, if there is one, or $\bot$ if there is none.
  We also maintain a mapping ($P(v,l)$) from level $l$ to a set of
  cross edges.  The set will contain all cross edges $e$ at level $l =
  l(e)$ that are incident on $v$. 
\end{itemize}
Table~\ref{table:symbols} summarizes the data structures and other
notation.  Figure~\ref{fig:parallelhyperedge} lists the pseudocode for
our algorithm.  The user-level functions are \ccode{insertEdges} and
\ccode{deleteEdges}, which insert and delete, respectively, batches of
edges.  Four of the functions are just for adding and removing
matched and cross edges.  They just update the relevant data
structures to maintain the invariants.  When adding a match we also
add all its sample edges.  When removing a match, we release its cross
edges and mark its vertices as free.   Removing a match assumes the
sample edges have already been converted to cross edges.

Note that we do not initialize $P(v,l)$ to empty for every level at the beginning of the algorithm (which would lead to $\Theta(n \log n)$ work as there are $O(\log n)$ levels). Instead, we store the id of each initialized bag in a hash table. Before accessing a bag, we check whether it has been initialized, and initialize it if it has not yet been initialized. 

Before describing how we handle batch insertions and deletions, we remark that our algorithm is quite different from Ghaffari and Trygub \cite{GT24}. Some of the differences include:
\begin{enumerate}[leftmargin=*]
\item Our leveling structure has a gap of 2, rather than $\Theta(r)$.
\item Our leveling structure permits matched edges that gain many cross edges to remain at a low level, instead of forcibly resettling at a higher level. 
\item In our algorithm, edges take on all ownership responsibilities: (matched) edges, not vertices, own edges.  
\item We process matched deletions at all levels at once, instead of correcting levels from high to low. 
\item We use random static maximal matching to randomly select matched edges, instead of iteratively choosing independent sets of edges to match via coin flips with increasing probability of success.
\end{enumerate}

Our algorithm most resembles Assadi and Solomon's hyperedge algorithm~\cite{AS21},
although, beyond being parallel, it differs in important ways. For an extended discussion of the differences between our algorithm and prior work, see Section~\ref{sec:compare}.

Inserting a batch of edges first identifies which of the edges are
free (are not incident on a matched edge) and runs a maximal matching
on these.  This could be any maximal matching.  It then adds the
matched edges to the leveled matching structure, each with just
themselves as the sample.  This places them on level $0$ since $\lg 1
= 0$.  Any remaining edges are added as cross edges.

Deleting a batch of edges has to separate the edges that are matched
from those that are not.  For those that are not we need to update the
data structures to remove them.  These unmatched edges could either be
sample edges or cross edges, but both types are easy to remove.  The
bulk of the algorithm is for processing the deleted matched edges.  To
handle matched edges we first remove them from their
sample space because they no longer exist in the graph and should not
be reinserted. Then we call \ccode{deleteMatchedEdges}, which is the
workhorse of the algorithm.

The function \ccode{deleteMatchedEdges} first takes all the sample
edges owned by the matched edges to be removed and inserts them as
cross edges.  We note that when inserting them as cross edges they
might or might not be owned by the same matched edge.  The algorithm
then separates the matched edges into two kinds: light and heavy.\footnote{Note that we handle light and heavy edges differently for efficiency (our work bounds), not for correctness. For example, if we designated all edges to be light, the matching would still be maximal. }  
The
light matches are ones for which the number of cross edges it owns is
not much larger than its sample size, and in particular for a match
$e$, less than $4r^22^{l(e)}$.  We allow ourselves to charge work
proportional to the sample size to each match, and hence can afford to
rematch these owned edges directly.  
The algorithm therefore simply
removes the light matched edges using \ccode{removeMatch}, which also
removes all their owned cross edges.   
The algorithm then reinserts the cross edges
with \ccode{insertEdges}.\footnote{In Solomon's algorithm this is
  referred to as deterministic settling.}  
  The \ccode{insertEdges} function
will find new matches for the removed cross edges, if possible.

We are now left with deleting the heavy matches.  We cannot afford to
find matches for the cross edges owned by the heavy matched edges without
taking new samples, which is performed by random settling.  Before
random settling, we remove from the structure all the heavy matches
along with their owned cross edges (using \ccode{removeMatch}).
Random settling is then performed in rounds.  One round, i.e., one
call to \ccode{randomSettle}, first takes all cross edges owned by
the heavy matches and runs a random greedy matching (\ccode{ParallelGreedyMatch}) on them.  
The
output of this matching will be a maximal set of matched edges, along
with a sample set for each.  Each matched edges $e$ is added to the
data structure with their sample $S_e$ at the appropriate level
($\lfloor \lg |S_e| \rfloor$) using \ccode{addMatch}. 

At this point there are a couple of invariants the random settle has
broken that need to be fixed.  Firstly, the levels of matches have
changed so some of the cross edges incident on the new matches may no
longer satisfy Invariant~\ref{ref:leveled}.\ref{invariant4}---i.e.
they might be owned by a match that is not at the maximum level of
their incident matched edges.  Secondly, since we selected new matches
among all edges owned by a match (not just the free ones), the new
matches might be incident on existing matches.  This would violate the
matching property.  We call these existing matches \emph{stolen}
matches since the algorithm will need to delete them.

Random settle fixes the cross edges by calling
\ccode{adjustCrossEdges}.  This removes all cross edges that might
change and then puts them back in with ownership on the maximum
incident level.  This makes use of the $P(v,l)$ data structure.  In
particular, for all the vertices that belong to the new matches it
collects all the cross edges up to one less than the level of the new
match.  These are cross edges that must be owned by the new match to
satisfy Invariant~\ref{ref:leveled}.\ref{invariant4}.  The cost to
transferring ownership can be charged to the new match the ownership
is assigned to.  If this cost charged to a match is significantly
larger than the sample size of the match, however, the match cannot
afford the charge.  In particular this is the case if the match $m$
now owns at least than $4r^22^{l(e)}$.  We refer to these as the
\emph{bloated} matches, and these will be deleted and their owned
edges random settled in the next round.

At the end of a round of random settling we therefore have two sets of
matched edges to be deleted: the bloated and the stolen edges.  We
delete these edges using \ccode{deleteMatchedEdges}.  As already
described, this will partition them into light and heavy.\footnote{The
bloated edges will always be heavy}  The light ones can be directly
removed and their owned edges reinserted with \ccode{insertEdges}.  The
heavy ones are removed and their owned edges are returned to be
matched in the next round of random settle.

The rounds of random settling continue until either a round is empty,
or two times the number of edges to process in a round is less than
the total number of edges processed in prior rounds.  In the second
case, we insert the remaining edges using \ccode{insertEdges}, and the
cost is charged against the prior work.  Because of the termination
condition, each round will double in size. Since the maximum size of a 
round is $m$, the total number of edges in the graph, there can be
at most $O(\log m)$ rounds.

\section{Analysis}
\label{sec:analysis}

We now analyze the cost bounds.  We consider both work and \depth{},
although the \depth{} analysis is straightforward and left until the end.
As in prior dynamic maximal matching
algorithms~\cite{BGS11,Sol16,AS21,GT24}, it is useful to consider the
lifetime of a match from when it is created until it is deleted.  This
is referred to as an \emph{epoch}.  It is also useful to distinguish
epochs that are deleted by the user in the call to
\ccode{deleteEdges}, vs. ones that are deleted internally by the
algorithm by a random settle.  The first kind are referred to as
\emph{natural deletions} of an epoch and the second as \emph{induced
deletions} of an epoch.  In the following we will equate the epoch $e$
with the edge $e$ corresponding to its match.

Roughly speaking, the work of the natural deletions of epochs can be
covered by the fact that on average about half the sampled edges will
be deleted by the user before a matched edge is selected.\footnote{We
intentionally do not say ``in expectation'' here, since the argument
is more subtle.}  However, the induced deletions of epochs
cannot be charged in this way since these epochs are deleted
prematurely, possibly immediately after they are created.  We
therefore charge induced deletions to the creation of new epochs, and
ultimately to a natural deletion.  We will therefore consider two work
costs to associate with epochs.  The first is the \emph{direct work}
which is associated with every epoch $e$ and covers all the
work for deleting matched edges.  Furthermore all work of each batch
operation on deleting matched edges can be associated with epochs
created and deleted during that operation.  The second is the
\emph{total work}, which is only associated with natural
epochs, and charges all the direct work of induced deletions to the
natural deletions.  We will show that the expected direct work for each epoch $e$ (natural or induced) 
is $O(r^3 2^{l(e)})$ and the expected total work for each natural epoch $e$ is $O(r^3 2^{l(e)})$. 

For each step in the algorithm, our charging scheme deterministically charges it to an epoch (or set of epochs). However the amount of work charged to an epoch is in expectation because the subroutines of the algorithm (static maximal matching, groupBy) themselves are randomized and have expected work bounds. Thus, many of the lemmas that follow have work bounds in expectation.

\subsection{Direct Work: Charging to all Epochs}

We start by analyzing the direct work that we charge to each epoch.
Here we only consider the work when the user deletes matched edges.
The work for adding edges and deleting unmatched edges is just $O(r)$
per edge and is left to the end.

We partition the direct work into three components: the work
associated with the light edges when deleting matched edges (the
\emph{light work}), the work associated with random settling (the
\emph{heavy work}), and the work for the final \ccode{insertEdges}
(the \emph{final work}).  The light work is charged to the deletion of
epochs while the heavy and final work are charge to the creation of
epochs.  Specifically, the light work includes: all work in
\ccode{insertEdges} when called from \ccode{deleteMatchedEdges}, as
well as the work of converting sample edges to cross edges, and all
work on light edges in \ccode{deleteMatchedEdges}.  The final work
just includes the final call to \ccode{insertEdges} in
\ccode{deleteEdges}.  All other work is heavy work.  This includes the
work in random settle and the work in \ccode{deleteMatchedEdges}
associated with heavy edges.  We first analyze \ccode{insertEdges}
since it is part of the light work.

\begin{lemma}
  \label{lemma:insert}
  The expected work of \ccode{insertEdges}$(E)$ is $O(r |E|)$.
\end{lemma}
\begin{proof}
  The random greedy match does $O(\sum_{e \in E} |e|) \le O(r|E|)$
  expected work.  Determining if a vertex is free, adding a match, and adding a
  cross edge each require $O(r)$ work, and there are $O(|E|)$ such
  operations for a total of another $O(r|E|)$ work.
\end{proof}

\begin{lemma}
  \label{lemma:light}
  The \emph{light work} can be covered by charging $O(r^3 \gap^{l(e)})$ in expectation to
  each deleted match $m \in \ccode{light}$ in \ccode{deleteMatchedEdges}.
\end{lemma}
\begin{proof} The work in \ccode{deleteMatchedEdges} for converting the sample edges
  to cross edges is $O(r)$ per edge. Each level $l$ light edge
  has at most $O(2^{l(e)})$ edges in its sample, for $O(r 2^{l(e)})$ charge to each light edge.

  The input edges $E$ to \ccode{insertEdges} when called from
  \ccode{deleteMatchedEdges} are the union of the owned edges of light
  edges.  Each $e \in \mbox{light}$ contributed at most $O(r^2
  \gap^{l(e)})$ edges to $|E|$ since it was light.  By the previous lemma the expected work of \ccode{insertEdges} is $O(r|E|)$, or $O(r)$ per edge.
  Since $O(r^2 \gap^{l(e)})$ edges each charge $O(r)$ to their deleted light matched edge, a deleted light matched edge in total gets charged  $O(r^3 \gap^{l(e)})$.
\end{proof}

The random settle proceeds in rounds and charging goes across
adjacent rounds so we consider the full cost across all rounds.

\begin{lemma}
  \label{lemma:heavy}
  The \emph{heavy work} across all rounds of a batch update can be covered by
  charging $O(r^3 \gap^{l(e)})$ in expectation to each newly matched epoch on an edge
  $e$ during the batch update.
\end{lemma}
\begin{proof}
  The cost of each round of random settle is either covered by matches
  in the current round or matches in the next round.  Specifically,
  all but raising of matches edges that end
  up being bloated, and the cost of \ccode{deleteMatchedEdges} at the
  end, will be covered by the current round.

  The total sample size in a round (with input $E$) is $|E|$ since all edges are in a
  sample.  The random greedy match does $O(r |E|)$ expected work by
  Theorem~\ref{lemma:greedymatchcost}.  The cost of adding matches is $O(r)$
  per edge in the sample sets ($S_e$) and hence is $O(r |E|)$.  The
  work for determining which incident edges need to be deleted is also
  $O(r|E|)$.  The cost of raising matches is proportional to the
  number of cross edges that are raised.  Each matched edge $e$ that
  does not become bloated will collect at most $O(r^2 2^{l(e)})$,
  and each such edge will have cost $O(r)$ for a total of $O(r^3
  2^{l(e)})$ work associated with $e$.

  The cost for raising edges that become bloated will be charged
  to the next round.  In particular, if there are $E_{next}$ edges in
  the next round, at most $|E_{next}|$ cross edges will be raised and become bloated,
  since each will contribute to the next round.  Each such edge contributes
  $O(r)$ work for $O(r |E_{next}|)$ total.  The cost of
  \ccode{deleteMatchedEdges} is split between light and heavy work,
  as discussed before.  In particular any of the light edges are
  charged to light work.  The cost of processing heavy
  edges is at most $O(r |E_{next}|)$ in expectation since each such edge and its
  owned edges are passed to the next round.

  We also need to cover the work of the initial call to \ccode{deleteMatchedEdges}
  in  \ccode{deleteEdges}. Any work on light edges
  is charged to light work, and any work on heavy edges is covered by
  matching the edges $E$ passed to the first round of random settle.

  All charges except for raising matches that are light have cost
  of $O(r |E_i|)$ charged against some round $i$.  This is $O(r)$ per
  edge in the sample sets of new epochs, and hence associates $O(r
  2^{l(e)})$ with each new epoch on edge $e$.  The cost is
  therefore dominated by the $O(r^3 2^{l(e)})$ work associated with
  raising the matched edge $e$ when it ends up light.
\end{proof}

\begin{lemma}
  \label{lemma:final}
  The \emph{final work} can be covered by charging $O(r \gap^{l(e)})$ in expectation to each
  newly matched epoch during a batch update.
\end{lemma}
\begin{proof}
  The total number of edges in the samples up to the final step is
  \ccode{sampledEdges}.  The \ccode{insertEdges} at the end on input $E$ requires
  $O(r|E|)$ expected work by Lemma~\ref{lemma:insert}. Since by the termination condition
  we have that $2|E| < \ccode{SampledEdges}$, the final insert cost
  can be covered by the prior sampled edges during the rounds
  of random settle.
\end{proof}

Putting the work costs together we get the following work bound.
\begin{lemma}
  All work in \ccode{deleteEdges} to delete matched edges can be
  covered by charging $O(r^3 \gap^{l(e)})$ in expectation to each start or end of an
  epoch during the batch.
\end{lemma}
\begin{proof}
  Follows from Lemma~\ref{lemma:light}, \ref{lemma:heavy}, and
  \ref{lemma:final}, which combined cover all work.
\end{proof}

\subsection{Total Work: Charging to Natural Epochs}

We have so far been able to charge all work in the algorithm to epochs
at a rate of $O(r^3 2^{l(e)})$ per epoch $e$, whether natural or
induced.  We now need to charge the work associated with induced
deletions of epochs to natural deletions of epochs.  This will be
based on a counting argument where we will show that the total number
of edges that are in the sample of natural deletions is at least a
constant fraction as many as the total number in induced deletions.
This will allow us to cover the cost of the induced deletions.

There are two sources of induced deletions, both caused by a random
settle.  Firstly, when we randomly match an edge $e$, there can be
edges that are already matched and incident on $e$ (\ccode{stolen} in
the algorithm) and need to be deleted.  We refer to these as
\emph{stolen deletes}.  Then there are the deletes that are due to the
fact that after raising an edge it might be too heavy for its level
(\emph{bloated deletes}). In the analysis we cover the stolen deletes
in the current round as well as the bloated deletes from the previous
round with new matches in the current round.

We say that the \emph{deleted sample size} of a round is the total
sample size of stolen deletes from the current round and bloated deletes
from the previous round.  We say that the \emph{added sample size} of
a round is the total sample size of newly matched edges in the current
round.

\begin{lemma}
  On each round of \ccode{randomSettle} with deleted sample size
  $S_d$ and added sample size $S_a$, $S_a \geq 2 S_d$.
  \label{lemma:adddelete}
\end{lemma}

\newcommand{\mold}{M_{\mbox{\footnotesize old}}}
\newcommand{\mnew}{M_{\mbox{\footnotesize new}}}
\newcommand{\es}{\;\;}

\begin{proof}
  The input edge set $E$ to \ccode{randomSettle} is the union of the
  owned cross edges of a set of heavy deleted matches
  $\mold$.  Each newly matched edge $m \in \mnew$ can therefore, for
  accounting purposes, be associated with the match in $\mold$ that
  previously owned it, and we will use $p(m)$ to indicate this deleted
  match.  Each matched edge $m \in \mnew$ can contribute up to $r-1$
  stolen deletes, one for each of its vertices except at least one
  incident on $p(m)$.  Now each $e \in \mold$ can have at most $r$
  newly matched edges incident on it (otherwise it would not be a
  matching).  Therefore we can associate $r(r-1)$ stolen deletes with
  each $e \in \mold$.  Furthermore we also associate the one bloated
  delete that might have caused each $e \in \mold$ to rise, and hence
  there are fewer than $r^2$ deletes associated with deleted edge $e$.
  This covers all the deletes.

  We make the following observations:
  \begin{enumerate}
  \item $|E| = \sum_{e \in \mnew} |S(e)| = \sum_{e \in \mold} |C(e)|$, since $|E|$ is the union of the cross edges in $\mold$ and all contribute to samples in $\mnew$,
  \item for $e \in \mold: |C(e)| \geq 4 r^2 2^{l(e)}$, since each is
    bloated or heavy, and
  \item $S_d < r^2 \sum_{e \in \mold} 2^{l(e) + 1}$, since each $e$ is
    associated with fewer than $r^2$ incident deletes and
    by Invariant~\ref{ref:leveled}.\ref{invariant4} the level $l$ of each
    such delete is at most $l(e)$.
    \end{enumerate}
Together this gives:
  \[
    S_a = \sum_{e \in \mnew} |S(e)|
        =  \sum_{e \in \mold} |C(e)|
        \geq  4 r^2 \sum_{e \in \mold} 2^{l(e)}
         >  2 S_d
\]
  \end{proof}

We note that having the gap between levels be a constant factor
($\alpha = 2$ in our case) instead of proportional to $r$ (as in
Assadi and Solomon~\cite{AS21}) is important to the proof.  In
particular, in the last step of the proof of Lemma \ref{lemma:adddelete} we lose a factor of $\alpha$. Since we have $\alpha=2$, this
brings the coefficient from $4$ to $2$ for us, which is fine. But if $\alpha$ was $\Theta(r)$, then the coefficient would
be reduced from $4$ to $\Theta(4/r)$, which is too small to be useful. We could have compensated for $\alpha=\Theta(r)$ by increasing the 
cutoff for being heavy (in \ccode{isHeavy}) by a factor of $r$, to $\Theta(r^3 2^{l(e)})$, but then the amortized cost of the overall algorithm would also have increased by
a factor of $r$, to be $O(r^4)$ per edge update. Thus, we save a factor of $r$ by using thinner levels.
In AS, reducing the level gap from $\Theta(r)$ to $2$ is not necessary because they assign every new match to some
specific deletes. This extra cost consideration in our algorithm comes from ``shuffling'' the sampled edges during the static matching.

\begin{lemma} \label{Lemma:5.7}
  If we start and end with an empty graph, all work done by deleting
  matched edges can be covered by charging each natural
  deletion $O(r^3 2^{l(e)})$ work in expectation.
\end{lemma}
\begin{proof}
  Let $S_a$ be the total sample size of all new epochs, $S_i$ the
  total sample size of all induced deletions and $S_n$ of natural
  deletes.  From Lemma~\ref{lemma:adddelete} we have that in each
  round of random settle the number of added samples is at least twice
  the number of samples from induced deletions.  At the end of a batch
  update, the algorithm runs a \ccode{insertEdges} on the remaining
  unmatched edges.  These are therefore not covered by the sample set
  of a random settle.  Due to the termination condition, however, we
  know that the sample size of previously matched edges is twice as
  large as the number of the remaining edges. Let $|\mbox{bloated}|$ be the sample 
  size of bloated remaining edges. Because these edges
   are bloated, they accrued a factor of at least $2r^2$ cross edges after rising. 
   Thus, the sample size of deleted matches from the last round of random 
   settle is at most $\frac{|\mbox{bloated}|}{2r^2} \le \frac{|\mbox{bloated}|}{2}$.  
   Therefore, the total sample space of induced
  deletions is at most the $\frac{1}{4} S_a$ from the final round plus
  the $\frac{1}{2} S_a$ from all other rounds, giving $S_a >
  \frac{4}{3} S_i$ within each batch deletion and hence also across
  all batch deletions.  Since we start and end empty, we have that
  $S_a = S_i + S_n$.  Putting these together we get $S_n > \frac{1}{3} S_i$.

  Finally we note that the direct work we assign to each epoch is
  based on its level not its sample size.  In particular, it is $c r^3
  2^{l(e)}$ for some constant $c$.  Hence the cost per
  sample $C_s$ satisfies $\frac{1}{2} c r^3 < C_s \leq c r^3$, since
  $2^{l(e)} \leq |S(e)| < 2^{l(e) + 1}$.  This implies the per sample
  cost we assign to induced and natural epochs differs by at most a
  factor of two, which further implies that the cost assigned to
  natural deletions is at least a constant fraction of the cost
  assigned to induced deletions.  The natural deletions can therefore
  cover the cost of induced deletions with a constant factor overhead.
\end{proof}

\subsection{Overall Cost Bounds \label{sec:overall-cost} }

We now consider the overall work and \depth{} bounds. We want to bound the total amount of natural sample space destroyed, so that we can bound the overall work with Lemma \ref{Lemma:5.7}. We use $\Phi$, as introduced in Section \ref{Section:3.1}, as a measure of sample space destroyed. We demonstrate that $\Phi$ is well-defined in the wider batch-dynamic setting. At a high level, $\Phi$ still functions properly in the wider algorithm because sampled edges are isolated while their matched edge lives. Then, we use $\Phi$ to bound the natural sample space deleted across the entire run of the algorithm. 


  Firstly we note that if we have a bound for the work from a start
  where the graph is empty to an end when it is empty again, this implies
  a bound that is at least as good for any prefix at least half as dense (in total batch size) as the whole sequence.   All runs are
  a prefix of at least half of an empty-to-empty sequence since we can just delete the edges and
  the number of final deletes is at most the number of prior operations.
  Hence we need only consider empty-to-empty runs to asymptotically
  bound the cost for the prefix.

Let $Q$ be the set of natural epochs. Let $d_1,d_2,\ldots,d_T$ be the set of deletions by the user across all batches. 
Suppose that we dynamically assign prices to edges and payments $\Phi$ to edge deletions, as described in Section \ref{Section:3.1}, as follows. Whenever we call \ccode{parallelGreedyMatch}, we assign prices to the edges involved. When an edge $e$ is deleted, we pay $\Phi(e)$. Then, if $e$ was sampled and unmatched, we decrement 1 from the price of $e$'s associated matched edge. For periods of time when an edge $e'$ has no price, we set $\Phi(e')=0$.



If a matched edge gets deleted in the same time step as an unmatched edge in its sample space, we will consider both deletes to be early. Accordingly, if a matched edge and $z$ of its sampled edges are deleted in batch, we will subtract $z$ from the matched edge's price before paying its price. 

\begin{lemma} Assigning prices, as in Section \ref{Section:3.1}, is well-behaved:
\begin{enumerate}[leftmargin=*]
\item For all times $t$, $\Phi$ is a function. 
\item An edge is sampled iff $\Phi(e) > 0$. An edge deletion is early iff the edge is sampled.
\item (Extension of Lemma \ref{Lemma:3.4}): Suppose that the user deletes edge $e$ at time $t$. Then $\E[\Phi(e)] \le 2$. 
\end{enumerate} \label{Lemma:price-scheme}  \end{lemma}

\begin{proof} We proceed as follows.

\begin{enumerate}[leftmargin=*]

\item Note that for an edge $e$, if $e$ has a price at time $t$, then by construction $e$'s matched edge still exists in the graph, so $e$ is sampled at time $t$. Note that a sampled edge is owned by exactly one epoch within the instance of \ccode{parallelGreedyMatch} that sampled it. Because sampled edges do not participate in the algorithm until their matched edge is deleted, this edge will not join the sample space of any other epochs.  Therefore, $e$ has at most one price at time $t$. Therefore, $\Phi(e)$ has at most one value. Because we set $\Phi(e)$ to 0 if an edge $e$ has no price (is a cross edge), we have that $\Phi(e)$ will map to exactly one value for all edges $e$. Therefore $\Phi$ is well-defined in the dynamic setting.

\item Consider the early delete of an edge $e$. Recall that $e$'s matched edge is $p(e)$, and that $e$ is in the sample space $S_{p(e)}$. By definition of early, $p(e)$ is still matched. Therefore, $p(e)$'s sample space is still sequestered from the rest of the algorithm, and so $e$ is still in $p(e)$'s sample space. Therefore, $e$ is sampled. Similarly, let $e'$ be a sampled edge. Because $e'$ is sampled, its matched edge $p(e')$ is still alive, so the deletion of $e$ is early. 

Note that the scheme in Section \ref{Section:3.1} never assigns payment 0 to a deletion. We conclude that a deleted edge gives positive payment iff the edge is sampled and iff the deletion is early. 

\item Note that $\E[\Phi(e)]$ is equal to $Pr[\Phi(e) > 0] \E[\Phi(e) | \Phi(e > 0)] + Pr[\Phi(e)=0] (0)$, which is no more than $\E[\Phi(e) | \textrm{early} ]$, which by Lemma \ref{Lemma:3.4} is no more than $2$. 

\end{enumerate} \end{proof}

\begin{lemma} The total natural sample space destroyed is no more than the total payment given by all user deletes: $\sum_{e \in Q} |S_e| \le \sum_{t=1}^T \Phi(d_t)$. \label{lemma:sampleprice}  \end{lemma}

\begin{proof}  For this proof, we look at the algorithm after it has finished running. Note that a natural epoch $e$ is destroyed by a user deletion and contributes $|S_e|$ to the total natural sample space destroyed, where $|S_e|$ is the original size of $e$'s sample space.

Consider the price associated with this same natural epoch $e$, deleted at time $t$. We proceed in a similar fashion to the proof of Lemma \ref{Lemma:3.5}. Suppose that $e$ started with $|S_e|$ sample space, has $y$ sample space at the beginning of time $t$, and has $z$ of its sampled edges also deleted at time $t$, in the same batch. We pay $y-z$ for deleting $e$, and the deletions of its sampled edges contribute $(|S_e|-y) + z$ payment, for $|S_e|$ payment total. Because edges have at most one payment value associated to them, this $|S_e|$ payment is only associated with the natural epoch $e$. 
Therefore, the total payment is equal to the sum of the payments associated with each epoch.

Since for each natural epoch, the sample space lost is no more than the payment, the lemma follows. Note that the total payment can be greater than the total natural sample space lost because unmatched sampled deletes from induced epochs pay but do not destroy natural sample space. \end{proof}
 
\begin{theorem}
  For any sequence of batch updates starting and ending with an
  empty graph, and with $N$ edges added or deleted across all batches,
  the total work is $O(r^3 N)$ in expectation.
\end{theorem}
\begin{proof}

In this proof, we view the algorithm before the run begins. By Lemma \ref{lemma:sampleprice}, observe that $\E[\sum_{e \in Q} |S_e|] \le \E[\sum_{t=1}^T \Phi(d_t)]$. By Lemma \ref{Lemma:price-scheme}, observe that $\sum_{t=1}^T \E[\Phi(d_t)] \le \sum_{t=1}^T 2 = 2T$. By Lemma \ref{Lemma:5.7}, the total work of the algorithm is bounded by \begin{equation*} \sum_{e \in Q} C r^3 2^{l(e)} \le \sum_{e \in Q} 2C r^3 |S_e| \le 2Cr^3 (2T) \le O(r^3 N), \end{equation*} where $C$ is the constant hidden by Lemma \ref{Lemma:5.7}. \end{proof}

\begin{lemma}
The \depth{} for every batch update is bounded by $O(\log^3 m)$, where $m$ is the maximum number of edges at any point in the graph.
\end{lemma}
\begin{proof}
  The \depth{} for \ccode{randomSettle}, \ccode{insertEdges}, and
  \ccode{deleteMatchedEdges} is $O(\log^2 m)$.  Specifically, the
  only steps in these functions that has \depth{} more than $O(\log m)$
  is the random greedy match, which has \depth{} $O(\log^2 m)$ by
  Lemma~\ref{lemma:greedymatchcost}.  The \depth{} of \ccode{deleteEdges}
  is $O(\log^2 m)$ up to but not including the while loop (dominated
  by the \ccode{deleteMatchedEdges}).  The while loop will iterate at
  most $O(\log m)$ steps since $E'$ cannot grow larger than $E$, and
  in each iteration $|E'|$ is growing by a constant factor.  Since the
  while loop makes a call to \ccode{randomSettle}, which has \depth{}
  $O(\log^2 m)$, the total \depth{} is $O(\log^3 m)$. \end{proof}

\section{Conclusion}

In this work, we demonstrate that a maximal matching can be maintained in the parallel batch-dynamic setting against an oblivious adversary in $O(r^3)$ expected amortized work per update. We do so by using a black box random greedy maximal matching algorithm to handle unsettled edges at all levels at once, and arguing that we can charge the work in aggregate to natural epochs, which have small total weight.

An open question is whether the work can be improved to $O(r^2)$. One key bottleneck in our work is the handling of stolen deletes. Since deleting a matched edge can cause $r^2$ stolen deletes, edges must have an $r^2$ factor before we permit them to random settle. Since edge updates naturally take $O(r)$ time for hyperedges, requiring an $r^2$ factor before random settling is cost prohibitive in the $O(r^2)$ setting. It is open if it would be possible in a parallel algorithm to reduce the number of stolen deletes, or if an entirely different set of techniques would be needed.

\myparagraph{Acknowledgements} This work was supported by National Science Foundation grants CCF-2119352 and CCF-1919223.

\bibliographystyle{ACM-Reference-Format}
\balance
\bibliography{bibliography/strings,bibliography/main}

\end{document}